\documentclass[12pt,reqno]{amsart}
\usepackage{amsmath,bbm}
\usepackage{latexsym}
\usepackage{amsfonts}
\usepackage{amssymb}
\usepackage{color}
\usepackage{bbm,dsfont}
\usepackage{graphicx}
\usepackage{fullpage}

\newtheorem{proposition}{Proposition}
\newtheorem{theorem}{Theorem}
\newtheorem{lemma}{Lemma}
\newtheorem{corollary}{Corollary}

\theoremstyle{definition}

\newtheorem{example}{Remark}
\newtheorem{definition}{Definition}

\newcommand{\R}{\mathbb R} 
\newcommand{\C}{\mathbb C} 
\renewcommand{\S}{\mathbb S} 
\newcommand{\RP}{\R{\bf P}} 
\newcommand{\half}{\tfrac{1}{2}} 

\newcommand{\hi}{\mathcal{H}} 
\newcommand{\lh}{\mathcal{L(H)}} 
\newcommand{\sh}{\mathcal{S(H)}} 
\newcommand{\no}[1]{\left\|#1\right\|} 
\newcommand{\tr}[1]{\mathrm{tr}\left[#1\right]} 
\newcommand{\rank}[1]{\mathrm{rank}(#1)} 
\newcommand{\id}{\mathbbm{1}} 
\newcommand{\nul}{O} 

\newcommand{\prem}{\mathcal{P}} 
\newcommand{\Min}[1]{\mathfrak{m}\left[#1\right]} 

\newcommand{\Ao}{\mathsf{A}}
\newcommand{\Bo}{\mathsf{B}}

\newcommand{\os}{\mathfrak{S}} 
\newcommand{\oscomp}{\os^\perp} 
\newcommand{\osgen}[1]{\mathfrak{S}(#1)} 
\newcommand{\As}{\mathfrak{S}(\Ao)} 
\newcommand{\ABs}{\mathfrak{S}(\Ao,\Bo,\ldots)} 
\newcommand{\Ascomp}{\As^\perp} 
\newcommand{\Bs}{\mathcal{B}} 
\newcommand{\Bscomp}{\mathcal{B}^\perp} 

\newcommand{\CP}{{\mathbb{C}}{\bf P}}

\newcommand{\ra}{\rightarrow}
\newcommand{\1}{\mathbbm{1}}

\def\>{{\rangle}}
\def\<{{\langle}}
\newcommand{\be}{\begin{equation}}
\newcommand{\ee}{\end{equation}}
\newcommand{\bea}{\begin{eqnarray}}
\newcommand{\eea}{\end{eqnarray}}

\begin{document}

\title[Quantum Tomography under Prior Information]{Quantum Tomography under Prior Information}

\author[Heinosaari]{Teiko Heinosaari$^\sharp$}
\address{$\sharp$ Turku Centre for Quantum Physics, Department of Physics and Astronomy, University of Turku, Finland. Email: teiko.heinosaari@utu.fi.}

\author[Mazzarella]{Luca Mazzarella$^\downharpoonright$}
\address{$\downharpoonright$ Department of Information Engineering, University of Padova,
Via Gradenigo 6/B, 35131 Padova, Italy. Email: mazzarella@dei.unipd.it.}

\author[Wolf]{Michael M. Wolf$^\ddag$}
\address{$\ddag$ Department of Mathematics, Technische Universit\"at M\"unchen, 85748 Garching, Germany. Email: wolf@ma.tum.de.}

\begin{abstract}
We provide a detailed analysis of the question: how many measurement settings or outcomes are needed in order to identify an unknown quantum state which is constrained by prior information? 
We show that if the prior information restricts the possible states to a set of lower dimensionality, then topological obstructions can increase the required number of outcomes by a factor of two over the number of real parameters needed to characterize the set of all states. 
Conversely, we show that almost every measurement becomes informationally complete with respect to the constrained set if the number of outcomes exceeds twice the Minkowski dimension of the set.
We apply the obtained results to determine the minimal number of outcomes of measurements which are informationally complete with respect to states with rank constraints.
In particular, we show that the minimal number of measurement outcomes (POVM elements) necessary to identify all pure states in a $d$-dimensional Hilbert space is $4d-3-c(d) \alpha(d)$ for some $c(d)\in[1,2]$ and $\alpha(d)$ being the number of ones appearing in the binary expansion of $(d-1)$.
\end{abstract}

\maketitle

\section{Introduction}\label{sec:intro}

Quantum tomography aims at identifying states of a quantum system. 
In order to achieve this aim, measurement data is often supplemented by prior information. 
In this work we consider cases where prior information effectively reduces the dimensionality, i.e., the number of parameters which are necessary to characterize the state of a system. Physically, one may think of scenarios of interferometry, process tomography or parameter estimation, where one prepares the initial state which then evolves depending on a certain number of unknown parameters before one measures the final system. Effective reductions of the number of parameters can also be due to a constraining symmetry or fixed energy or particle number.

We are interested in determining the minimal number of measurement outcomes or measurement settings which are required for identifying a state taken from such a reduced set. Clearly, if states in the considered set are parameterized by a certain number of independent real parameters, then we need at least this number of measurement outcomes or binary measurements in order to pinpoint the state. As an example take the manifold of pure states in a $d$-dimensional Hilbert space. Their description requires $2d-2$ real parameters, as opposed to $d^2-1$ real parameters needed to describe an arbitrary density matrix. So if we want to determine a pure state by a single measurement with $m$ outcomes, $m\sim 2d$ seems necessary and $m\sim d^2$ appears to be achievable.  The question about the smallest possible $m$ in this particular example has been addressed in a number of publications \cite{Weigert92,AmWe99a,AmWe99b,Finkelstein04,FlSiCa05}, but the answer has remained somewhat elusive, so far.

A related problem has been addressed based on \emph{compressed sensing} ideas, where it has been shown \cite{GrLiFlBeEi10} that for $d\times d$ density matrices of rank $r$, $m=O(dr\log(d)^2)$ binary measurements are sufficient in order to identify the state with high probability.
In this light we emphasize that our focus lies on schemes which identify the unknown state unambiguously and deterministically. In this work we disregard questions regarding the robustness of the schemes, the complexity of the necessary post-processing of the measurement data or the ability to verify the assumed prior information.

\emph{Outline}. Typically, when we want to know the minimal number of measurement outcomes related to some subset of states, we need to independently consider upper and lower bounds. 
In Section \ref{sec:upper} we concentrate on upper bounds, while in Section \ref{sec:lower} we provide some methods to obtain lower bounds. 
The upper bounds are of geometric or algebraic nature and often come with concrete constructions. The lower bounds are mainly topological in nature. They are based on the observation that any measurement which is informationally complete when supplemented by prior information is a mapping into the space of measurement outcomes which preserves topological invariants.  In the case of pure states the upper and lower bounds essentially match which enables us for instance to show that for a minimal informational complete measurement   $m\sim 4d$ up to an additive logarithmic correction.
This exemplifies our general finding that, loosely speaking, on the one hand topological obstruction can force $m$  to be twice the dimension of a considered manifold, whereas on the other hand geometric reasoning allows us to show that  such an overhead is essentially always sufficient.

\emph{Notation}. In this paper $\hi$ is a fixed \emph{finite} dimensional complex Hilbert space.
We denote by $\lh$ the set of linear operators on $\hi$ and we will write $\no{L}$ for the operator norm which coincides with the largest singular value of $L\in\lh$.
A positive operator $\varrho\in\lh$ having trace one is a \emph{density operator}, also referred to as \emph{state}, and we denote by $\sh$ the set of all states on $\hi$.

\section{Quantum tomography prerequisites}\label{sec:os}

In this section we will introduce and analyze the linear algebra framework for quantum tomography. Later, in Section \ref{sec:lower}, we will then add a topological perspective.

The goal of quantum tomography is to identify an unknown quantum state from the statistics of measurement outcomes. We will consider tomographic schemes where a fixed (as opposed to an adaptive) measurement setting is chosen, and we will refrain from considering errors for instance caused by finite statistics.  
In particular, in the following \textquotedblleft~measurement~\textquotedblright always refers to a full statistical experiment rather than to a single-shot experiment. 

Our analysis deals with several different albeit closely related approaches:  1) single measurements with many outcomes; 2) several measurements with possibly fewer outcomes; 3) several measurements where only expectation values are considered.

Quantum measurements are generally described by \emph{positive operator valued measures} (POVMs) \cite{PSAQT82,OQP97}. Since we are eventually interested in the minimal required number of outcomes, we will restrict ourselves to POVMs with finitely many outcomes. For $n\in\mathbb{N}$ outcomes such a POVM is characterized by a collection of positive operators $\{\Ao_j\}_{j=1}^n:=\Ao\subset\lh$ satisfying $\sum_{j=1}^n \Ao_j = \id$. 
Notice that a POVM with $n$ outcomes is already determined by $n-1$ operators; the last operator $\Ao_n$ is given by $\Ao_n=\id-\sum_{j=1}^{n-1} \Ao_j$.

For a POVM $\Ao$ and a state (density matrix) $\varrho$, we denote by $\varrho^{\Ao}$ the corresponding probability distribution of measurement outcomes. It is given by the formula
\begin{equation*}
\varrho^\Ao(j) = \tr{\varrho\Ao_j} \, , \quad j=1,\ldots,n \, .                                                                                                                            
\end{equation*}
Again, due to normalization the last component is determined by the others via $\varrho^\Ao(n)=1-\sum_{j=1}^{n-1}\varrho^\Ao(j)$. 
A POVM with $n$ outcomes therefore induces a continuous, convex-linear map from $\sh$ into $\mathbb{R}^{n}$ (or $\mathbb{R}^{n-1}$ if we disregard the somewhat irrelevant last component). 
In the most favorable case $\varrho_1^\Ao \neq \varrho_2^\Ao$ for all pairs of different states $\varrho_1$ and $\varrho_2$; in this case $\Ao$ is called \emph{informationally complete} \cite{Prugovecki77}, \cite{Busch91} and the respective map $\varrho\mapsto\varrho^\Ao$ is injective on $\sh$.
In this work we are interested in cases where prior information, or premise, limits the possible initial states of the system.
In other words, we want to identify $\varrho$ not necessarily among all states but within some subset $\prem\subseteq\sh$ of states. 
This motivates the following definition.

\begin{definition}[$\prem$-informational completeness]
Let $\prem\subseteq\sh$ be a subset of density matrices.
A POVM $\Ao$ is called \emph{informationally complete w.r.t. $\prem$} if $\varrho\mapsto\varrho^\Ao$ is injective on $\prem$, i.e., 
\begin{equation*}
 \forall \varrho_1,\varrho_2\in\prem:\quad \varrho_1^\Ao = \varrho_2^\Ao \Rightarrow\varrho_1=\varrho_2.
\end{equation*}
\end{definition}

In a  more general scheme one may perform several measurements and use all of their measurement outcome statistics.
Let $\{\Ao,\Bo,\ldots\}$ be a finite collection of POVMs. 
We denote by $\ABs$ the linear span of the set 
$$
\{\Ao_1,\ldots,\Ao_n\} \cup \{\Bo_1,\ldots,\Bo_m\} \cup \cdots \, ,
$$
where $n$ is the number of measurement outcomes of $\Ao$, $m$ is the number of measurement outcomes of $\Bo$, and so forth. 
In particular, 
\begin{equation*}
\As := \Bigl\{ \sum_j c_j \Ao_j \mid c_j \in \C \Bigr\} \, .
\end{equation*}
The set $\ABs$ is a linear subspace of $\lh$, the vector space of all linear operators on $\hi$.
Clearly, $\id\in\ABs$ and $\ABs^\dagger=\ABs$, where $\ABs^\dagger=\{A^\dagger:A\in\ABs\}$.
A linear subspace of $\lh$ with these two properties is called an \emph{operator system}.

If the system is described by a state $\varrho$ and the measurement statistics of the collection $\{\Ao,\Bo,\ldots\}$ is given, then we can calculate the expectation $\tr{\varrho X}$ of any operator $X\in\ABs$. 
Therefore, from the point of view of informational completeness, only the linear span $\ABs$ is relevant.
Two collections $\{\Ao,\Bo,\ldots\}$ and $\{\Ao',\Bo',\ldots\}$ are considered equivalent from this perspective if their linear spans are the same.
The following simple result shows that any operator system is generated by a single POVM.
Therefore,  we can restrict our investigation to the situation where only a single POVM is measured. 
   
\begin{proposition}[POVMs and operator systems]\label{prop:operator-system}
Let $\os\subseteq\lh$ be an operator system (i.e. $\os$ is a linear subspace such that $\id\in\os$ and $\os^\dagger=\os$).
There exists a POVM $\Ao$ such that $\os=\As$ and $\Ao$ has $\dim\os$ outcomes.
Any POVM $\Bo$ satisfying $\os=\osgen{\Bo}$ has at least $\dim\os$ outcomes.
\end{proposition}

\begin{proof}
Every operator $X\in\os$ can be written as a linear combination of two selfadjoint operators, 
\begin{equation*}
X=\frac{1}{2}(X+X^\dagger) + i \frac{1}{2i} (X-X^\dagger) \, , 
\end{equation*}
which belong to $\os$.
If $X\in\os$ is selfadjoint, then $X$ can be written as a linear combination of two positive operators, 
\begin{equation*}
X=\frac{1}{2}(\no{X}\id+X) - \frac{1}{2} (\no{X}\id - X) \, ,
\end{equation*}
which again belong to $\os$. 
Therefore, we can choose a basis for $\os$ consisting of positive operators and the identity $\id$; let $\{E_1,\ldots,E_{m},\id\}$ be such.
For each $j=1,\ldots,m$, we set $\Ao_j=\frac{1}{m\no{E_j}}E_j$. 
The element $\Ao_{m+1}$ is defined as $\id-\sum_{j=1}^m \Ao_j$.
Then the collection $\Ao_1,\ldots,\Ao_{m+1}$ forms a POVM and $\os=\As$.

For any POVM $\Bo$, the generated operator system $\osgen{\Bo}$ cannot have more linearly independent elements than $\Bo$ has outcomes.
This implies the last claim.
\end{proof}

We denote by $\os^\perp$ the orthogonal complement of an operator system $\os$ with respect to the Hilbert-Schmidt inner product. 
In particular, for a POVM $\Ao=\{\Ao_j\}_{j=1}^n$ we have 
\begin{equation*}
\Ascomp = \left\{ B \in \lh : \ \tr{B^\dagger \Ao_j}=0 \ \forall j=1,\ldots n \right\} \, .
\end{equation*}
We notice that $\tr{M}=0$ for any $M\in\os^\perp$ (since $\id\in\os$) and $(\os^\perp)^\dagger=\os^\perp$ (since $\os^\dagger=\os$).
The latter property implies that $\os^\perp$ is spanned by its selfadjoint elements.
Since $\os\oplus\os^\perp=\lh$, we have
\begin{equation}
\dim\os + \dim\oscomp = d^2 \, .
\end{equation} 
The complement $\Ascomp$ of $\As$ is related to the informational completeness of $\Ao$ in the following simple way.

\begin{proposition}[Operator systems and informational completeness]\label{prop:comp-states}
Let $\os\subseteq\lh$ be an operator system and let $\prem\subseteq\sh$ be a set of states. 
Then a POVM $\Ao$ satisfying $\As=\os$ is informationally complete w.r.t. $\prem$ iff $$(\prem-\prem)\cap\os^\perp=\{0\}.$$ In other words, for any pair of states $\varrho_1, \varrho_2\in\prem$ the following are equivalent:
\begin{itemize}
\item[(i)] $\varrho_1^\Ao = \varrho_2^\Ao$,
\item[(ii)] $\varrho_1 - \varrho_2 \in \Ascomp$.
\end{itemize}
\end{proposition}

\begin{proof}
For all states $\varrho_1,\varrho_2\in\prem$, we have:
\begin{eqnarray*}
 \varrho_1^\Ao = \varrho_2^\Ao & \Longleftrightarrow & \tr{\varrho_1 A} = \tr{\varrho_2 A} \quad \forall A\in\os \\
& \Longleftrightarrow & \tr{(\varrho_1-\varrho_2) A} = 0 \quad \forall A\in\os \\
& \Longleftrightarrow & \varrho_1-\varrho_2 \in \Ascomp \, .
\end{eqnarray*}
\end{proof}

For illustrative purposes, let us use the above framework to confirm the well known fact that there exists an informationally complete POVM with $d^2$ outcomes and that $d^2$ is the minimal number for which informational completeness can be achieved for all of $\mathcal{S}(\mathbb{C}^d)$ \cite{CaFuSc02}. 
The set of all operators $\mathcal{L}(\mathbb{C}^d)$ is an operator system with dimension $d^2$.
By Proposition \ref{prop:operator-system} there exists a POVM  $\Ao$ with $d^2$ outcomes and satisfying $\As=\mathcal{L}(\mathbb{C}^d)$. 
Since $\Ascomp=\{0\}$, Proposition \ref{prop:comp-states} implies that $\Ao$ is informationally complete w.r.t. all states. 
On the other hand, if $\Bo$ is a POVM with less than $d^2$ elements, then $\osgen{\Bo}^\perp$ contains a nonzero selfadjoint operator $X$ (since $\osgen{\Bo}^\perp$ is generated by its selfadjoint part). 
The states $\frac{1}{d}\id$ and  $\frac{1}{d}(\id+\no{X}^{-1}X)$ are different but they can not be distinguished by $\Bo$.

In some tomography schemes one may only infer from the expectation values of measurements rather than from the entire measurement outcome statistics.
Such expectation values are characterized  by a collection of selfadjoint operators.
Let $\{S_1,S_2,\ldots,S_n\}$ be a finite collection of selfadjoint operators.
Again, if we know the expectations $\tr{\varrho S_1}, \ldots, \tr{\varrho S_n}$, we can calculate $\tr{\varrho X}$ for any operator $X$ belonging to the linear span of $S_1,\ldots,S_n$. 

\begin{definition}[$\prem$-informational completeness for selfadjoint operators]
Let $\prem\subseteq\sh$ be a subset of states.
A collection $\{S_1,S_2,\ldots,S_n\}$ of selfadjoint operators is called \emph{informationally complete w.r.t. $\prem$} if
\begin{equation*}
 \tr{\varrho_1 S_j}=\tr{\varrho_2 S_j} \ \forall j \Rightarrow \varrho_1=\varrho_2
\end{equation*}
for all $\varrho_1,\varrho_2\in\prem$.
\end{definition}

We are interested in questions of the type:
what is the minimal number of 
\begin{itemize}
\item[1)] POVM outcomes
\item[2)] selfadjoint operators 
\end{itemize}
needed to have an informationally complete measurement scheme w.r.t. to a given premise $\prem\subseteq\sh$?

Let us notice that the identity operator $\id$ does not give any useful information since $\tr{\varrho\id}=1$ for all states $\varrho$. 
Since an operator system $\osgen{\Ao}$ related to any POVM $\Ao$ contains the identity operator $\id$, the minimal number of POVM outcomes is higher than the minimal number of selfadjoint operators.
The above two questions 1) -- 2) are thus related in the following way.

\begin{proposition}\label{prop:min}
Let $\prem\subseteq\sh$. The following are equivalent.
\begin{itemize}
\item[(i)] There exists a POVM $\Ao$ with $n$ outcomes and $\Ao$ is informationally complete w.r.t. $\prem$.
\item[(ii)] There exists a set $\{S_1,S_2,\ldots,S_{n-1}\}$ of $n-1$ selfadjoint operators and $\{S_1,S_2,\ldots,S_{n-1}\}$ is informationally complete w.r.t. $\prem$.
\end{itemize}
\end{proposition}

\begin{proof}
(i)$\Rightarrow$(ii): Suppose $\Ao$ has $n$ outcomes and $\Ao$ is informationally complete w.r.t. $\prem$.
Since $\Ao_n=\id - \sum_{j=1}^{n-1} \Ao_j$, we conclude that the set $\{\Ao_1,\ldots,\Ao_{n-1}\}$ of selfadjoint operators is informationally complete w.r.t. $\prem$.
\newline
(ii)$\Rightarrow$(i): Suppose  $\{S_1,S_2,\ldots,S_{n-1}\}$ is a set of selfadjoint operators which is informationally complete w.r.t. $\prem$.
For each nonzero $S_j$, we define
\begin{equation*}
\Ao_j:=\big(\half \id + \half \no{S_j}^{-1} S_j\big)/(n-1) \, .
\end{equation*}
Then $\nul\leq\Ao_j\leq\id/(n-1)$ and by setting $\Ao_n:=\id - \sum_{j=1}^{n-1} \Ao_j$ we obtain a POVM which is informationally complete w.r.t. $\prem$.
\end{proof}

As a consequence of Proposition \ref{prop:min} we only need to calculate the minimal number of either selfadjoint operators or POVM elements. 

\begin{definition}(Minimal informationally complete measurements)
Let $\prem\subseteq\sh$ be any subset of states. 
We denote by $\Min{\prem}$ 
\begin{quote}
- the minimal number of selfadjoint operators which are informationally complete w.r.t. $\prem$.
\end{quote}
or, equivalently (by Prop. \ref{prop:min}),
\begin{quote}
- the minimal number of positive operators determining a POVM which is informationally complete w.r.t. $\prem$
(i.e., the number of POVM elements subtracted by one).
\end{quote}
\end{definition}
The remaining part of this work is devoted to deriving upper and lower bounds on $\Min{\prem}$ either for various concrete subsets $\prem$ or depending on specific properties of $\prem$ -- in particular its dimensionality and topology.

\section{Upper bounds}\label{sec:upper}

\subsection{States with bounded rank}\label{sec:rank}

As the first class of examples, we consider the subset $\prem_r:=\{\varrho\in\sh : {\rm rank}(\varrho)\leq r\}$ of all states whose rank is bounded by a number $r<d/2$ where $d=\dim{\mathcal{H}}$.
How many measurement outcomes suffice in order to identify a state taken from $\prem_r$?

\begin{theorem}[States with bounded rank]\label{prop:rank}
If $1\leq r<d/2$, then there exists a POVM $\Ao$ which is informationally complete w.r.t. $\prem_r$ and has $4r(d-r)$ outcomes.
Therefore, 
\begin{equation*}
\Min{\prem_r}\leq 4r(d-r)-1 \, .
\end{equation*}
\end{theorem}

\begin{proof}
We will construct a subspace $\Bs$ of $d\times d$ matrices with the properties that
\begin{itemize}
	\item[(a)] $\Bs^\dagger=\Bs$,
	\item[(b)] $\tr{B}=0$ for every $B\in\Bs$,
	\item[(c)] $\dim\Bs=(d-2r)^2$,
	\item[(d)] $\rank{B}\geq 2r+1$ for every nonzero $B\in\Bs$.
\end{itemize}
From (a) -- (b) it follows that $\os:=\Bscomp$ is an operator system.
By Prop. \ref{prop:operator-system} there exists a POVM  $\Ao$ such that $\As=\os$ and $\Ao$ has $\dim\os$ outcomes.
From (c) we get that $\Ao$ has $d^2-(d-2r)^2=4r(d-r)$ outcomes.
Finally, if $\varrho_1,\varrho_2\in\prem_r$, then $\rank{\varrho_1-\varrho_2}\leq 2r$.
By Prop. \ref{prop:comp-states} this implies that a POVM $\Ao$ is informationally complete w.r.t. $\prem_r$ if $\rank{B}\geq 2r+1$ for every nonzero operator $B\in\Ascomp$.
This is guaranteed by (d). Hence, constructing a subspace $\Bs$ with the properties (a) -- (d) will prove the proposition.

The main part of our construction follows \cite{CuMoWi08}.
The following fact will be needed.
Let $M$ be a totally nonsingular $m\times m$-matrix with positive entries.
(Recall that a matrix is called totally nonsingular if all of its minors are nonzero.)
For instance, a Vandermonde matrix of the form
\begin{equation*}
M=\left( \begin{array}{ccccc}
1 & \alpha_1 & \alpha_1^2 & \cdots & \alpha_1^{m-1} \\
1 & \alpha_2 & \alpha_2^2 & \cdots & \alpha_2^{m-1} \\
\vdots & \vdots & \vdots && \vdots \\
1 & \alpha_m & \alpha_m^2 & \cdots & \alpha_m^{m-1}
\end{array} \right)
\end{equation*}
with $0<\alpha_1<\alpha_2<\cdots<\alpha_m$ has strictly positive minors \cite{Fallat01}.
Since $M$ is totally nonsigular, any linear combination of $\ell$ columns of $M$ contains at most $\ell-1$ zero elements:
if some linear combination of $\ell$ columns of $M$ would contain $\ell$ zero elements, then we could find a $\ell\times\ell$ submatrix of $M$ with linearly dependent columns. 
This contradicts the requirement that all the minors of $M$ are nonzero. 

We will now define a set of $d\times d$ matrices that span the sought subspace $\Bs$.  
By the $k$th diagonal of a matrix $[M_{ij}]$ we mean the elements $M_{ij}$ with $i-j=d-k$.
(In other words, we label the diagonals from the lower left corner upwards. The main diagonal is then the $d$th diagonal.)

For each $k$ satisfying $2r+1\leq k \leq d-1$, we build up $k-2r$ matrices as follows.
We choose $k-2r$ columns from a totally nonsingular $k\times k$ -matrix and we put them to the $k$th diagonal and $0$'s elsewhere, hence obtaining $k-2r$ linearly independent matrices.
Any matrix $P$ which is a linear combination of these $k-2r$ matrices has at least $2r+1$ nonzero elements on the $k$th diagonal. Since all the matrix elements not in the $k$th diagonal are $0$, we see that the $P$ has rank at least $2r+1$.

We also take all transposes of the previously constructed matrices to our spanning set of matrices. In addition we construct and add to the set $d-2r$ diagonal matrices (with nonzero entries only on the $k$=$d$th diagonal). To this end, we choose $d-2r$ columns from a totally nonsingular $d\times d$ -matrix $M$.
Let $v_1,\ldots,v_{d-2r}$ denote these column vectors.
Next, we choose a real vector $u$ which has no zero entries and which is orthogonal to every $v_1,\ldots,v_{d-2r}$.
A possible choice is, for instance, the last row from the inverse matrix of $M$.
Then the new vectors $\tilde{v}_1,\ldots,\tilde{v}_{d-2r}$ defined as entrywise products of $v_1,\ldots,v_{d-2r}$ with $u$ each have the property that their components sum up to zero. Hence, from these vectors we can construct $d-2r$ traceless diagonal matrices such that again any non-zero linear combination of them has rank at least $2r+1$.

In total, we have built up $d-2r+2\sum_{k=2r+1}^{d-1}(k-2r)=(d-2r)^2$ linearly independent matrices.
The previously mentioned subspace $\Bs$ is the linear span of these matrices.
The properties (a)-(c) are immediate consequences of the construction.
To verify (d), suppose that $B\in\Bs$.
Let $k_B$ be the largest $k$ such that the $k$th diagonal of $B$ contains nonzero elements.
Then the $k_B$th diagonal contains actually $2r+1$ nonzero elements. 
The $(2r+1) \times (2r+1)$ submatrix having those $2r+1$ nonzero elements in its main diagonal is lower triangular, therefore has nonzero determinant.
It follows that the rank of $B$ is at least $2r+1$. 
\end{proof}

Let us remark that a recently introduced method based on compressed sensing uses $O(dr\log(d)^2)$ outcomes to identify a rank-$r$ state with high probability \cite{GrLiFlBeEi10}.
The number of outcomes given in Theorem \ref{prop:rank} therefore beats the compressed sensing approach. The latter, however, might be advantageous regarding the simplicity of the classical post-processing, robustness and verifiability of the assumption.

\subsection{Pure states}\label{sec:pic}

We will now have a closer look at the set $\prem_1$ of pure states.
For this case, Theorem \ref{prop:rank} implies that $\Min{\prem_1}\leq 4d-5$.
We will first provide some simple arguments showing that indeed $\Min{\prem_1}=4d-5$ for $d=2,3$.
In general, we will see later, based on topological reasoning, that the leading order $4d$ is the best possible.

We remark that Flammia et al. proved that a POVM which is informationally complete w.r.t. the set of pure states has at least $2d$ outcomes \cite{FlSiCa05}. 
They also constructed a POVM with $2d$ elements capable of distinguishing almost all (but not all) pairs of pure states. A similar construction was given by Finkelstein \cite{Finkelstein04}.

\begin{proposition}\label{prop:not-pic}
For a POVM $\Ao$, the following conditions are equivalent:
\begin{itemize}
\item[(i)] $\Ao$ is not informationally complete w.r.t. the set of pure states.
\item[(ii)] $\Ascomp$ contains a selfadjoint operator $T\neq 0$ with $\rank{T}\leq 2$.
\item[(iii)] $\Ascomp$ contains a selfadjoint operator $T$ with $\mathrm{rank}(T)=2$.
\end{itemize}
\end{proposition}

\begin{proof}
(i)$\Rightarrow$(iii): Suppose $\Ao$ is not informationally complete w.r.t. pure states. 
By Prop. \ref{prop:comp-states} there exists two pure states $\varrho_1\neq\varrho_2$ such that $\varrho_1 - \varrho_2=:T \in \Ascomp$. 
The operator $T$ is selfadjoint and $\rank{T}= 2$.
\newline
(ii)$\Leftrightarrow$(iii) is seen by observing that there is no traceless selfadjoint operator with rank 1. 
\newline
(iii)$\Rightarrow$(i)
Suppose there is a selfadjoint operator $T\in\Ascomp$ with rank 2. 
Since $\tr{T}=0$ it has two nonzero eigenvalues $\pm\lambda$. The operator $T':=\frac{1}{\lambda} T$ then has spectral decomposition $T'=P_1-P_2$, where $P_1$ and $P_2$ are one-dimensional projections, hence pure states.
By Prop. \ref{prop:comp-states} the POVM $\Ao$ cannot distinguish $P_1$ and $P_2$.
\end{proof}

From Proposition \ref{prop:not-pic} we conclude the following simple characterization.

\begin{corollary}[Pure state informationally complete measurements]\label{cor:pic}
A POVM $\Ao$ is informationally complete w.r.t. the set of pure states if and only if every nonzero selfadjoint operator $T\in\Ascomp$ has $\rank{T}\geq 3$. 
\end{corollary}

\begin{example}[The qubit case]
It is an immediate consequence of Corollary \ref{cor:pic} that in the qubit case (i.e. $\dim\hi=2$) a POVM $\Ao$ is informationally complete w.r.t. pure states if and only if $\Ascomp=\{0\}$. 
Therefore, in the qubit case informational completeness for pure states implies informational completeness for all states.  
One can also easily see this by a direct inspection of the Bloch sphere.
\end{example}

Note that if $\Ascomp$ contains an operator $M$ with $\rank{M}=1$, then $M+M^\dagger$ is selfadjoint and $\rank{M+M^\dagger}\leq 2$.
Thus, $\Ao$ is not informationally complete w.r.t. pure states if $\Ascomp$ contains a rank-1 operator.
On the other hand, if $\Ascomp$ contains an operator $M$ with $\rank{M}=2$, this does \emph{not} imply that there exists a selfadjoint operator $T$ with $\rank{T}=2$. 
For instance, the subspace
\begin{equation*}
	\left\{ 
	\begin{pmatrix}
0 & 0 & \alpha & 0 \\
0 & 0 & 0 & \alpha \\
\beta & 0 & 0 & 0 \\
0 & \beta & 0 & 0
\end{pmatrix}
: \alpha,\beta \in \C \right\}
\end{equation*}
contains rank-2 matrices but no rank-2 selfadjoint matrix.

As an application of Corollary \ref{cor:pic}, we can easily characterize all POVMs in dimension $3$ which are  informationally complete w.r.t. the set of pure states.

\begin{proposition}[Pure state informational completeness in dimension $3$]\label{prop:PIC-dim3}
A POVM $\Ao$ is informationally complete w.r.t. the set of all pure states in $\mathcal{S}(\mathbb{C}^3)$  if and only if it falls into one of the two classes:
\begin{itemize}
	\item $\Ascomp=\{0\}$ (i.e. $\Ao$ is informationally complete w.r.t. all states).
	\item $\Ascomp=\{cT:c\in\C\}$ for some invertible selfadjoint operator $T$ with $\tr{T}=0$.
\end{itemize}
In particular, a minimal POVM which is informationally complete w.r.t. pure states has $8$ outcomes.
\end{proposition}

\begin{proof}
We first prove the following: if $\dim\As\leq 7$, then $\Ao$  cannot distinguish all pairs of distinct pure states.
We need to show that there exists a singular selfadjoint operator $T\in\Ascomp$.
The claim then follows from Cor. \ref{cor:pic}.
From $\dim\As\leq 7$ follows that $\dim\Ascomp \geq 2$.
There thus exist two linearly independent selfadjoint operators $X,Y\in\Ascomp$.
If either $X$ or $Y$ is singular, we are done. 
So let us assume that $\det(X)>0$ and $\det(Y)<0$ (otherwise we redefine $X\to -X$ or $Y\to-Y$). By the intermediate value theorem the function $t\in\mathbb{R}\mapsto \det\big(tX+(1-t)Y\big)$  has to be zero for some $t=t_0\in (0,1)$. The corresponding operator $t_0X+(1-t_0)Y$ is thus a singular element of $\As^\perp$.

We conclude that if $\Ao$ is informationally complete for the set of all pure states, then either $\dim\As=8$ or $\dim\As=9$.
In the latter case $\Ao$ is informationally complete for all states. 
If $\dim\As=8$, then $\dim\Ascomp=1$ and $\Ascomp$ is thus generated by a single selfadjoint operator $T$.
From Cor. \ref{cor:pic} follows that $\Ao$ is informationally complete w.r.t. all pure states exactly when $T$ is invertible. 
\end{proof}

Proposition \ref{prop:PIC-dim3} shows how to construct POVMs that are informationally complete w.r.t. pure states for $d=3$:
fix any invertible selfadjoint operator $T$ with $\tr{T}=0$. Then the complement space is an operator system from which the sought POVM can be obtained as a spanning set (see the proof of Proposition \ref{prop:operator-system}).

\begin{example}[Necessity of prior information]
The characterization given in Proposition \ref{prop:PIC-dim3} helps to illustrate one possible drawback of POVMs which are merely informationally complete w.r.t. to  a subset of states: even if a POVM $\Ao$ can distinguish any pair of distinct pure states, it may not be capable of distinguishing pure states from mixed states. 
In other words, we may not be able to verify the premise from the measurement outcome statistics.

To give an example, let $\Ao$ be a POVM in $\C^3$ with $\Ascomp=\{cT:c\in\C\}$ for some invertible selfadjoint operator $T$ with $\tr{T}=0$.
Suppose $T$ has a single positive eigenvalue (if not, take $-T$).
Let $T=\lambda_1 P_1-(\lambda_2 P_2 + \lambda_3 P_3)$, $\lambda_j>0$, be the spectral decomposition of $T$ into one-dimensional projections $P_j$.
The POVM $\Ao$ cannot distinguish the pure state $P_1$ from the mixed state $\varrho=\lambda_2/\lambda_1 P_2 + \lambda_3/\lambda_1 P_3$.
We conclude that by measuring $\Ao$ one cannot identify the pure state $P_1$ from the measurement outcome statistics \emph{without} making use of the prior knowledge that the state is pure.
\end{example}

We already noted that Theorem \ref{prop:rank} leads to the general upper bound $\Min{\prem_1}\leq 4d-5$ albeit without yielding an explicit set of operators.
Now, we give another proof of this result via a simple direct construction.
The construction is inspired by a topological embedding given by James in \cite{James59}.
Consider two types of matrices $X_\alpha$ and $Y_\beta$  which we label by integers $\alpha=1,\ldots,2d-2$ and $\beta=1,\ldots,2d-3$ respectively. 
The $X_\alpha$'s are taken to be such that $\big(X_{\alpha}\big)_{kl}=\delta_{k+l,\alpha+1}$, i.e., there are $1$'s along the $\alpha$'th anti-diagonal and zeroes elsewhere. 
The $Y_\beta$'s are similarly defined with nonzero entries solely along the anti-diagonals, in this case $\big(Y_{\beta}\big)_{kl}=0$ unless $k+l=\beta+2$. 
The entries are chosen such that the matrices are anti-symmetric with entries $i$ below the diagonal.  The constructed matrices thus have the following structure:
\be
X \sim\left(\begin{array}{ccccc} & & &1&\\ & & 1 & &\\ & 1 & & &\\ 1 & & & &\\ & & & &\\
\end{array}\right), \quad \mbox{and}\quad Y\sim\left(\begin{array}{ccccc} & & &-i&\\ & & -i & &\\ & i & & &\\ i & & & &\\ & & & &\\
\end{array}\right).
\ee

\begin{theorem}[Pure state informational completeness--explicit construction]\label{prop:james}
The set $\{S_j\}:=\{X_\alpha,Y_\beta\}$ consisting of $4d-5$ selfadjoint operators is informationally complete w.r.t. the set of all pure states on $\mathbb{C}^d$.
\end{theorem}

\begin{proof}
Following \cite{James59} we use the following type of auxiliary matrices. 
Consider a set of upper-triangular matrices $\{C_\gamma\in\C^{d\times d}\}_{\gamma=2,\ldots,2d}$ which are such that  $\big(C_\gamma\big)_{kl}=0$ if $k+l>\gamma$ and $\big(C_\gamma\big)_{kl}\neq0$ if $k+l=\gamma$.
We claim that for all vectors $x,y\in\C^d$ the following holds:
\begin{eqnarray}\label{eq:JamesLemma}
&& \textrm{if}\quad \<x|C_\gamma|x\>=\<y|C_\gamma|y\> \ \textrm{for all}\ \gamma \, , \\ \nonumber
&& \textrm{then}\ y=e^{i\varphi}x \ \textrm{for some}\ \varphi\in\R \, .
\end{eqnarray}
To prove this claim, assume that the $n$'th component of $x$ is the first with a non-zero entry $x_n\neq 0$. 
Then $\<x|C_{2n}|x\>=\big(C_{2n}\big)_{n,n}|x_n|^2$ which by the hypothesis in Eq.\eqref{eq:JamesLemma} implies that there is a $\varphi\in\R$ so that $y_l=e^{i\varphi}x_l$ holds for all $l\leq n$ (since $x_l=y_l=0$ for $l<n$ and $|x_n|=|y_n|$). By induction we can now prove that the same proportionality has to hold for all other components. So assume for some $m\geq n$ that  $y_l=e^{i\varphi}x_l$ holds for all $l\leq m$. Then for $\gamma=m+n+1$
\bea \<x|C_\gamma|x\> &=& \sum_{i+j\leq\gamma}\big(C_\gamma\big)_{ij}\bar{x}_i x_j \\
&=& \big(C_\gamma\big)_{n,m+1}e^{i\varphi}\bar{y}_n x_{m+1} +\!\!\!\!\sum_{i+j\leq m+n} \big(C_\gamma\big)_{ij}\bar{y}_i y_j\nonumber,
\eea
where we replaced $x_l\ra e^{-i\varphi}y_l$ for all $l\leq m$ and exploited that $C_\gamma$ is upper triangular and that $x_l=0$ for all $l<n$. Together with the hypothesis in Eq.\eqref{eq:JamesLemma} and the assumption that $\big(C_\gamma\big)_{n,m+1}\neq 0$ this implies indeed that $y_{m+1}=e^{i\varphi}x_{m+1}$. 

We now exploit Eq.\eqref{eq:JamesLemma} for specific matrices which we construct as $C_{2d}=\1$ and for $\gamma=2,\ldots,2d-1$ as
\be \big(C_\gamma\big)_{kl}=\left\{\begin{array}{ll}
                                     \delta_{k+l,\gamma}, & \ k<l, \\
                                     1/2, & \ k=l=\gamma/2 \\
                                     0, & \ \mbox{otherwise.}
                                   \end{array}\right.\ee
Note that $C_\gamma=(X_{\gamma-1}+i Y_{\gamma-2})/2$ for $\gamma=3,\ldots,2d-1$ and $C_2=X_1/2$. Hence, if $\<x|S_j|x\>=\<y|S_j|y\>$ for every $j$, then $\<x|C_\gamma|x\>=\<y|C_\gamma|y\>$ for every $\gamma=2,\ldots,2d-1$. 
The remaining condition $\<x|C_{2d}|x\>=\<y|C_{2d}|y\>$ holds due to $\no{x}=\no{y}$.
Therefore it follows from Eq.\eqref{eq:JamesLemma} that the set $\{S_j\}$ is informationally complete w.r.t. the set of all pure states.
\end{proof}

It is possible to slightly improve the upper bound $\Min{\prem_1}\leq 4d-5$ at the cost of having to deal with a more complicated set of operators introduced in the work of Milgram \cite{Milgram67}:

\begin{theorem}[Pure state informational completeness -- improved bound]\label{prop:milgram}
Let $\alpha$ denote the number of $1$'s in the binary expansion of $d-1$.
There exists a collection of $m$ selfadjoint operators which is informationally complete w.r.t. pure states in $\mathbb{C}^d$, if
\be m=\left\{\begin{array}{ll} 4d-3-\alpha &\ \mbox{for odd }d,\\
4d-4-\alpha &\ \mbox{for even } d\geq 4.
\end{array}\right.\ee
\end{theorem}

For odd $d$ this upper bound can be worse than the previously derived $4d-5$, but for most dimensions it is below $4d-5$.
Notice that $\alpha$ satisfies $1\leq\alpha\leq \log_2(d)$. 
We will see later that this improvement is, in fact, nearly optimal. 
With this refined upper bound we can calculate the exact value of $\Min{\prem_1}$ for small dimensions $d$ in Sec. \ref{sec:geometry}.

\begin{proof}
The construction of the operators is based on the work of Milgram \cite{Milgram67}. 
We will only argue why this corresponds to a proper measurement scheme rather than reproducing the rather cumbersome construction.

For ${m}$ as in the proposition, Milgram constructed a set of ${m}$ bilinear maps, i.e., matrices $A_j\in\C^{d\times d}$, $j=1,\ldots {m}$ which have the following properties:

(i) \emph{Vanishing real inner product} in the sense that for all $x\in\C^d$ we have $\<x,A_j x\>_\R=0$ for the real inner product $\<x,y\>_\R:=\half(\<x,y\>+\overline{\<x,y\>})$. That is, each $A_j$ has to be skew-symmetric w.r.t. the real inner product and thus anti-selfadjoint w.r.t. to the standard complex inner product, i.e., $A_j^\dagger=-A_j$. In order to see the latter, note that every matrix $A$ can be written as a sum $A=A^s + A^a$, where $A^s=\half(A+A^\dagger)$ is selfadjoint and $A^a=\half(A-A^\dagger)$ is anti-selfadjoint.
For each $x\in\C^d$, we see that $\<x,A^s x\>$ is real and $\<x,A^a x\>$ is imaginary.
Therefore, the condition $\<x,Ax\>_\R=0$ for all $x\in\C^d$ is equivalent to $A=A^a$.

(ii) \emph{Completeness.} We can define matrices $T_j=iA_j$, which according to (i) are selfadjoint, and as noted by Mukherjee \cite{Mukherjee81} such that the map $f:\C^d\ra\R^m$ defined via $f(x)_j=\<x,T_j x\>$ has the property that $f(x)=f(y)$ implies that $x$ is proportional to $y$.

The set of ${m}$ selfadjoint matrices $T_j$ therefore leads to a measurement scheme which is informationally complete w.r.t. the set of pure states.
\end{proof}

\subsection{Generic bounds depending on fractal dimension}\label{sec:fractal}

In this part we will discuss general upper bounds on $\Min{\prem}$ which depend only on the dimensionality of $\prem$. To this end, it is useful to regard any measurement scheme as a linear map between real Euclidean vector spaces. More precisely, we identify the set of selfadjoint operators on $\C^d$ with $\R^{d^2}$ and the set of $m$-tuples of selfadjoint operators in $\mathcal{L}(\C^d)$ as linear maps from $\R^{d^2}$ into $\R^m$.

Let us assume that $\prem$ is a \emph{closed} subset of $\sh$.
Then $\prem$ can be identified with a compact subset of $\R^{d^2}$.
The \emph{Minkowski dimension} (also called \emph{box dimension}) $D(\prem)$ of $\prem$ is obtained by considering the minimal number $N_\epsilon(\prem)$ of $\epsilon$-balls needed to cover $\prem$ and taking the limit 
\begin{equation*}
D(\prem):=\lim\sup_{\epsilon\ra 0}\frac{\log\big(N_\epsilon(\prem)\big)}{\log(1/\epsilon)}.\end{equation*}
For instance, if $\prem$ is a smooth manifold of real dimension $d(\prem)$, then $D(\prem)=d(\prem)$. 

By Man\'e's theorem \cite{Mane81,HuKa99} for a compact set $\prem$ almost any (in the Lebesgue measure sense) linear map $\Lambda$ from $\R^{d^2}$ into $\R^{m}$ is injective on $\prem$ if $m>2D(\prem)$. 
The injectivity of $\Lambda$ is clearly equivalent to the property that the corresponding set of $m$ selfadjoint operators is informational complete w.r.t. $\prem$.
Hence, we conclude with the following result.

\begin{theorem}[Informational completeness of generic measurements]\label{prop:mane}
Let $\prem\subseteq\sh$ be a closed subset of the set of density matrices with Minkowski dimension $D(\prem)$.
Then almost any (in the Lebesgue measure sense) collection of $m>2D(\prem)$ selfadjoint operators is informationally complete w.r.t. $\prem$. In particular, $\Min{\prem}\leq 2D(\prem)+1$.
\end{theorem}

\begin{example}
In principle, this bound can be refined to $m>\delta_{\prem-\prem}$, where $\delta_{\prem-\prem}$ is the \emph{Hausdorff dimension} of the set $\prem-\prem:=\{\varrho_1-\varrho_2|\varrho_i\in\prem\}$ \cite{Robinson09}.
 This bound is generally better since $\delta_{\prem-\prem}\leq D(\prem-\prem)\leq 2D(\prem)$, but it may be more difficult to handle. 
We also mention that for the inverse mappings H\"older continuity can be proven and the respective constants can be bounded \cite{BeEdFoNi93,HuKa99}.
\end{example}

Let us apply Theorem \ref{prop:mane} to the case of pure states. 
This is a smooth manifold of real dimension $2d-2$, hence $D(\prem_1)=2d-2$. 
Thus, almost any collection of $4d-3$ selfadjoint operators is informationally complete w.r.t. pure states.

A related example for which it might be difficult to obtain good bounds by other means is the set of depolarized pure states. Let $\sigma\in\mathcal{S}(\C^d)$ be any mixed state and define $\prem_\sigma:=\{\varrho\in\mathcal{S}(\C^d)|\varrho\in\lambda\sigma+(1-\lambda)\prem_1,\;\lambda\in[0,1] \}$. Then $D(\prem_\sigma)=D(\prem_1)+1$ so that almost any set of $4d-1$ selfadjoint operators is informationally complete w.r.t. $\prem_\sigma$.

\section{Lower bounds}\label{sec:lower}

The main idea which in this section provides lower bounds on $\mathfrak{m}$ is that an informationally complete measurement preserves the topology of $\prem$ when we regard the measurement scheme as a mapping from $\prem$ into the real vector space $\R^\mathfrak{m}$ corresponding to measurement outcomes or probabilities. More precisely, on its image in $\R^\mathfrak{m}$ this map has to be a homeomorphism or, with some additional assumptions, a diffeomorphism. For various manifolds $\prem$ the existence of such maps, i.e., the possibility of a (differential) topological embedding into $\R^\mathfrak{m}$ is well studied and allows us to translate non-embedding results into lower bounds on $\mathfrak{m}$. Fortunately, in some cases of interest these bounds are very close to or even match the upper bounds obtained in the previous section.

\subsection{Measurements as homeomorphism}\label{sec:topology}

As before we will identify the set of selfadjoint operators on $\C^d$ with $\R^{d^2}$. In this way, we can regard the set $\sh$ of all states as well as any of its closed subsets $\prem\subseteq\sh$ as a compact subset of $\R^{d^2}$ from which it inherits a natural topology.

Let $\Ao$ be a POVM with $m+1$ outcomes.
This induces a map
\begin{equation}\label{eq:h}
h_\Ao:\sh\to\R^{m} \, , \quad h_\Ao(\varrho) := (\tr{\varrho \Ao_1},\ldots, \tr{\varrho \Ao_{m}}) \, ,
\end{equation}
which is injective as a map from $\prem$ iff $\Ao$ is informationally complete w.r.t. $\prem$. With a slight abuse of notation we write $h_\Ao$ for the map from $\prem$ as well as for the extended map from $\R^{d^2}$. 

\begin{proposition}[Informational completeness and topological embeddings]\label{prop:emb}
Let $\prem\subseteq\sh$ be a closed subset. A POVM $\Ao$ is informationally complete w.r.t. $\prem$ iff the map $h_\Ao$ is a topological embedding of $\prem $ into $\R^m$. 
\end{proposition}
\begin{proof}
By definition a topological embedding is an injective continuous map which has a continuous inverse on its image, i.e., a homeomorphism onto its image. So injectivity of $h_\Ao$ and thus informational completeness of $\Ao$ is implied by $h_\Ao$ being a topological embedding. For the converse note that $h_\Ao$ is linear on $\R^{d^2}$ and thus continuous. Moreover, by assumption $h_\Ao:\prem\rightarrow h_\Ao(\prem)$ is a continuous bijection and as such has a continuous inverse since $\prem$ is  compact.
\end{proof}

The usefulness of Proposition \ref{prop:emb} is that it gives a method to derive a lower bound for $\Min{\prem}$.
If we know that $\prem$ does not admit a topological embedding into $\R^m$, then $\Min{\prem} > m$.

\begin{figure}
\includegraphics[width=0.7\textwidth]{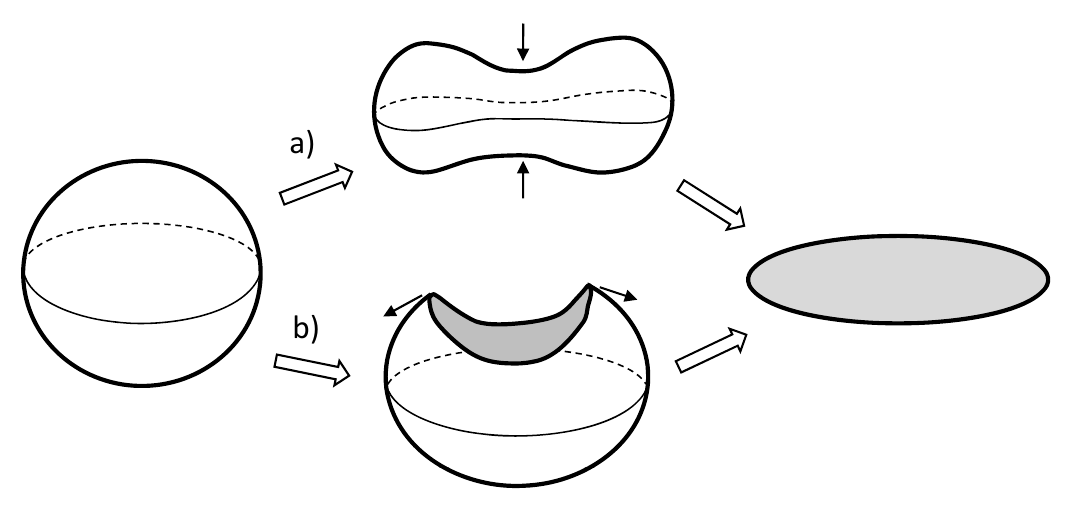}
\caption{\label{fig:Bloch} A mapping from the sphere $\S^2\subset\R^3$ to the plane $\R^2$ is either (a) not injective, or (b) not continuous. 
When regarding the sphere as the manifold of pure qubit states and the mapping as a measurement with $3$ outcomes (of which only two can have independent probabilities), this simple picture implies that a pure state informationally complete measurement requires $4$ outcomes.}
\end{figure}

\begin{example}[Pure qubit states revisited]
As a simple demonstration of Proposition \ref{prop:emb}, let us consider the subset of all pure qubit states.
This subset is homeomorphic to the unit sphere $\S^2\subset\R^3$.
Hence, any POVM with $3$ outcomes defines a continuous map from $\S^2$ into $\R^2$. 
A measurement separating all pairs of distinct pure states has to map different initial states (=points on the sphere) to different points in $\R^2$. 
A discontinuous mapping (option (b) in Fig.\ref{fig:Bloch}) cannot arise from a measurement. Every continuous map into the plane, however, necessarily identifies different points of the sphere (option (a) in Fig.\ref{fig:Bloch}). 
In fact, by the \emph{Borsuk-Ulam theorem} every continuous map from $\S^{n-1}$ into $\R^n$ maps some pair of antipodal points to the same point.
Consequently, we recover the fact that any pure state informationally complete POVM on a qubit must have more than $3$ outcomes. 
\end{example}

\begin{example}[Measurements on several copies]
Although we do not discuss this scenario in other parts of this work, let us mention that Proposition \ref{prop:emb} holds as well if we allow global measurements on multiple copies, i.e., if we replace $\varrho$ by $\varrho^{\otimes n}$ for some $n\in\mathbb{N}$. The map $h_\Ao$ then becomes non-linear but it is evidently still continuous, which is all we need. In this way, also the aforementioned qubit results generalizes to the multiple copy case. 
\end{example}

By a \emph{surface} we mean a 2-dimensional topological manifold, i.e., a topological space with the property that every point has a neighborhood which is homeomorphic to an open subset of $\R^2$.
A surface is orientable or non-orientable.
(We recall that a possible characterization of this distinction is that a surface is non-orientable if and only if it contains a homeomorphic image of the M\"obius strip.)

\begin{corollary}[2-manifolds]\label{cor:surface}
If $\prem\subset\sh$ is a closed surface without boundary, then $\Min{\prem}\geq 3$. Moreover,  if $\prem$ is 
 non-orientable in addition, then $\Min{\prem}\geq 4$. 
\end{corollary}
\begin{proof} This follows from Prop.\ref{prop:emb} by noting that (i) a compact surface without boundary cannot be embedded into $\R^2$ (where any compact set does have a boundary) and (ii) by the classification of surfaces, non-orientability does not permit an embedding into $\R^3$.
\end{proof}
We note that this type of purely topological reasoning cannot give better lower bounds since by  Whitney's embedding theorem every surface can be embedded in $\R^4$ \cite{Whitney44}.

\begin{figure}
\includegraphics[width=0.5\textwidth]{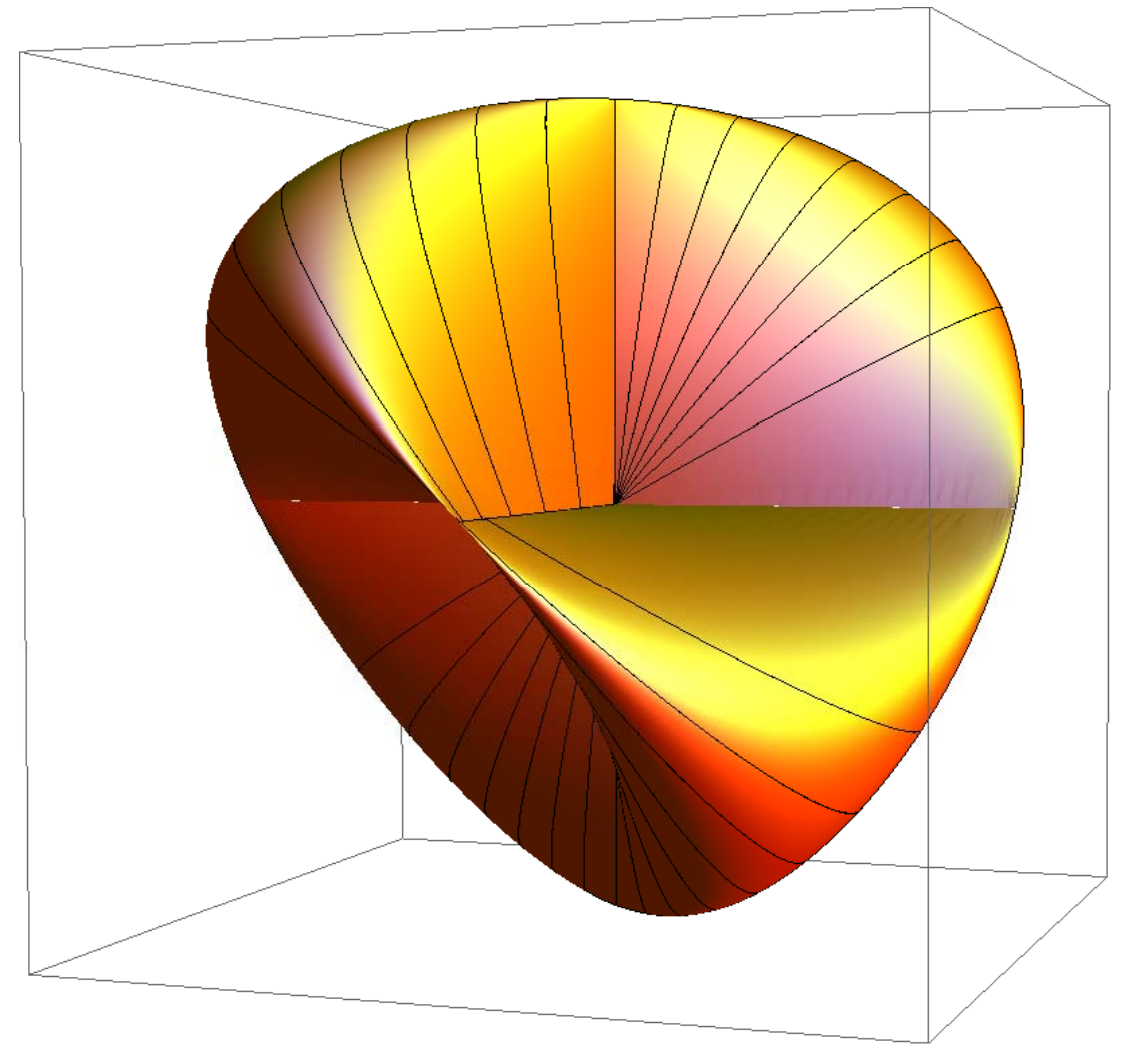}
\caption{\label{fig:Boy} The \emph{Roman surface} is a self-intersecting surface in $\mathbb{R}^3$ obtained by a mapping from the real projective plane $\mathbb{R}{\bf P}^2$. In our context it arises when the manifold of three-dimensional real, pure states is mapped onto the expectation values of three observables. Since the manifold is not orientable, every three-outcome measurement has to be non-injective---here reflected by the self-intersections of the surface. }
\end{figure}

As an application of Corollary \ref{cor:surface}, let us consider the set of pure states in $\C^3$ with real amplitudes. That is, we fix an orthonormal basis $\{\varphi_j\}_{j=1}^3$ and the pure states under investigation correspond to the vectors $\psi = \sum_{j=1}^3 r_j \varphi_j$, $r_j \in \R$.
\begin{corollary}[Pure states with real amplitudes in $\C^3$]\label{cor:real3}
Let $\prem\subset\mathcal{S}(\C^3)$ be the set of pure states with real amplitudes w.r.t. a fixed basis in $\C^3$. Then $\Min{\prem}=4$. 
\end{corollary}
\begin{proof}
Due to normalization every state in $\prem$ can be represented by a unit vector $x\in \S^2$. 
Since $x$ and $-x$, however, represent the same state, we have to identify antipodes so that $\prem$ is homeomorphic with $\RP^2$, the \emph{real projective plane}. 
Since $\RP^2$ is non-orientable, it follows from Cor. \ref{cor:surface} that $\Min{\RP^2}\geq 4$.
In order to see that in fact $\Min{\RP^2}=4$ consider the map $x\mapsto(x_1x_2,x_2x_3,x_3x_1,x_1^2-x_2^2)$. 
This is a topological embedding of $\RP^2$ in $\R^4$ which can be realized by a measurement scheme; the four components are expectation values of the selfadjoint matrices (written in the fixed orthonormal basis $\{\varphi_j\}_{j=1}^3$)
\begin{align}\label{eq:sigmas}
\half \begin{pmatrix}
0 & 1 & 0  \\
1 & 0 & 0  \\
0 & 0 & 0
\end{pmatrix} \, , \quad 
\half \begin{pmatrix}
0 & 0 & 0  \\
0 & 0 & 1  \\
0 &1 & 0
\end{pmatrix} \, , \quad 
\half \begin{pmatrix}
0 & 0 & 1  \\
0 & 0 & 0  \\
1 & 0 & 0
\end{pmatrix} \, , \quad 
\begin{pmatrix}
1 & 0 & 0  \\
0 & -1 & 0  \\
0 & 0 & 0
\end{pmatrix} \, , \quad 
\end{align}
respectively.
\end{proof}

 \begin{example}[Roman surface]
 The first three matrices in \eqref{eq:sigmas} give rise to measurement results which form the \emph{Roman surface} displayed in Fig.\ref{fig:Boy}. 
The failure of informational completeness (due to disregarding the necessary fourth measurement) is reflected by self-intersections of the surface.
 \end{example}

\subsection{Measurements as diffeomorphisms}\label{sec:geometry}

Manifolds of interest in quantum tomography often have a differentiable structure -- they are \emph{smooth manifolds}. 
In such a case we may resort to differential topology, which imposes more restrictive conditions on the existence of \emph{smooth embeddings}. Before we apply these conditions to the concrete cases of pure states and states with general rank constraints we provide some general background.

Suppose that $\prem$ is a smooth manifold. 
A smooth embedding $f$ of $\prem$ into $\R^m$ is a smooth map which is a topological embedding (i.e. homeomorphism onto its image) and has the property that the derivative of $f$ is everywhere injective.
Although for most smooth manifolds the minimal embedding dimension $m$ is not known exactly, quite narrow intervals have been determined for many cases of interest (see \cite{James71,Adachi93} for an overview). One tool to derive lower bounds on the minimal embedding dimension $m$ is Chern's result \cite{Chern48} that a smooth embedding into $\R^m$ requires that the dual Stiefel-Whitney classes vanish, i.e., $\bar{W}(\prem)_i=0$ for all $i\geq m-D(\prem)$. 
Other bounds can be obtained from an index theorem due to Atiyah and Hirzebruch \cite{AtHi59} and similar ideas \cite{Mayer65,Sugawara79}.
 On the positive side a general upper bound is due to Whitney \cite{Whitney44} who showed that a smooth embedding $\prem\ra\R^m$ always exists if $m\geq 2 D(\prem)$ (actually $m\geq 2 D(\prem)-1$ unless $m$ is a power of 2). 
 Whitney's bound is known to be optimal, i.e., in the worst case the dimension of the Euclidean space has to be twice the dimension of the manifold.

Again, for a POVM $\Ao$ with $m+1$ outcomes, we denote by $h_\Ao:\sh\to\R^{m}$ the induced mapping defined in Eq.(\ref{eq:h}).
In order to apply the known results on the lower bounds on dimensions for smooth embeddings, we have to show that for a $\prem$-informationally complete POVM $\Ao$, the related map $h_\Ao:\prem\rightarrow h_\Ao(\prem)$ is a diffeomorphism.
In Subsection \ref{sec:topology} we saw already that it is a homeomorphism. 

Throughout we will suppose that $\prem$ is a \emph{compact embedded submanifold} of $\R^{d^2}$ where we identify the latter with the space of selfadjoint matrices in $\C^{d\times d}$. With a slight abuse of notation we will  write $\prem$ for both, the manifold and its inclusion in $\R^{d^2}$.   
We denote by $T_p(\prem)$ the tangent space of $\prem$ at $p\in\prem$ and by $h_*:T_p(\prem)\ra T_{h_\Ao(p)}\big(h_\Ao(\prem)\big)$ the derivative, which is a linear map between the tangent spaces (sometimes call pushforward). 
The cone defined by
$$
\Delta(\prem):=\{X\in\R^{d^2} | X=\lambda(M_1-M_2) \ \textrm{for some}\ M_i\in\prem, \lambda>0\} \, .
$$
 will play an important role in the following.

\begin{theorem}[Smooth embeddings]\label{thm:infoembeddingsmooth} 
Let $\prem$ be a \emph{compact embedded submanifold} of $\R^{d^2}$, where the latter is identified with the space of selfadjoint matrices in $\C^{d\times d}$.  Suppose that $h_\Ao:\prem \ra\R^m$ is a mapping of the form in Eq.\eqref{eq:h}.  Then $h_\Ao(\prem)$ is a smooth embedding of $\prem$ in $\R^m$ if $h_\Ao$ is injective on $\prem$ and for all $p\in\prem$ the following inclusion holds: \be T_p(\prem)\subseteq\Delta(\prem).\label{eq:diffintan}\ee
 \end{theorem}
\begin{proof}
To show that $h_\Ao$ is a smooth embedding we need to prove (i) that it is smooth, which follows from linearity, (ii) that it is a topological embedding, which follows from the assumed injectivity and Prop. \ref{prop:emb} and (iii) that it has an injective derivative everywhere. 

Due to the linearity of $h_\Ao$ on $\R^{d^2}$ we have $h_*=h_\Ao$ but we have to be careful with the domains in order to argue that the injectivity of $h_\Ao$ (as a mapping from $\prem$) implies the injectivity of $h_*$ (as a set of mappings from $T_p(\prem)$ for any $p\in\prem$). By assumption, for any $p\in\prem$ and $X\in T_p(\prem)$ we have $X\in\Delta(\prem)$. Then indeed $h_*(X)=0$ together with the injectivity of $h_\Ao$ implies $X=0$ since $h_*(X)=\lambda\big(h(M_1)-h(M_2)\big)$ is zero only if $M_1=M_2$.
\end{proof}

We will now prove the inclusion \eqref{eq:diffintan} for the subset of states in $\mathcal{S}(\C^d)$ which are proportional to a projection of rank $r\leq d$, i.e., states which are maximally mixed within a subspace of dimension $r$. 
This set forms a smooth manifold of real dimension $2r(d-r)$ which is isomorphic to the \emph{complex Grassmannian} manifold $G(r,d-r)$ \cite{Dimitric96}. 
$G(1,d-1)$ is the set of pure states.

\begin{lemma}\label{lem:Grass} 
The inclusion $T_p(\prem)\subseteq\Delta(\prem)$ holds for all $p\in\prem$ if $\prem$ is the complex Grassmannian manifold $G(r,d-r)$ understood as the submanifold in the space of $d\times d$ selfadjoint matrices which consists of all orthogonal projections of rank $r$.
 \end{lemma}

\begin{proof}
Let us first identify the tangent space at an arbitrary point $P\in\prem$ which is now a selfadjoint projection with  $\tr{P}=r$. 
Considering a curve within $\prem$ through $P$ given by the unitary orbit $c(t):=e^{iHt}Pe^{-iHt}$ for some selfadjoint matrix $H$ and $t\in\R$. 
The derivative 
\begin{equation*}
\partial_t c(t)\big|_{t=0}= i[H,P]
\end{equation*}
 is an element of $T_p(\prem)$ and in fact, such derivatives span the entire tangent space \be \label{eq:tanspaceGrass} T_p(\prem) =\big\{X=X^\dagger|X=i[H,P]\mbox{ for some }H=H^\dagger\big\}.\ee In order to see this we have to show that they span a vector space which has the same dimension as the manifold (for which $D(\prem)=2r(d-r)$). To this end, note that there is a one-to-one relation between commutators and block off-diagonal matrices in the sense that we can always write
 \be i[H,P]=\left(
              \begin{array}{cc}
                0 & C \\
                C^\dagger & 0 \\
              \end{array}
            \right),\quad C\in\C^{r\times (d-r)},\label{eq:tanGrassoff}\ee
 in the basis where $P=\1\oplus 0$. So the dimensions match, which verifies Eq.(\ref{eq:tanspaceGrass}).

In a suitable basis any element $X\in T_P(\prem)$ is such that \be X=\left[\bigoplus_{i=1}^r\left(
                                                                                 \begin{array}{cc}
                                                                                    0& c_i \\
                                                                                    c_i& 0 \\
                                                                                 \end{array}
                                                                               \right)\right]\oplus 0_{d-2r},\quad c_i\geq 0,\ee
since Eq.(\ref{eq:tanGrassoff}) allows us to work with the singular values $\{c_i\}$ of $C$ by transforming $X\mapsto (U\oplus V)X(U\oplus V)^\dagger$ with appropriate unitaries $U$ and $V$. Setting $\lambda:=\max_i{c_i}$ equal to the operator norm of $X$ we complete the proof if we show that every $2\times 2$ matrix of the form $c\sigma_x$ with $c\in[0,1]$ is a difference of two projections. This can seen to be true by taking the difference of two pure qubit states whose Bloch vectors are parameterized by $(c,\pm \sqrt{1-c^2},0)$.
\end{proof}

A special case is the $2d-2$ dimensional manifold of pure states in $\C^{d}$ which can be identified with the \emph{complex projective space} $\CP^{d-1}$.  
The map from $\CP^{d-1}$ to selfadjoint rank-one projections is itself a smooth embedding. 
This together with Theorem \ref{thm:infoembeddingsmooth} and the above Lemma implies that non-embedding results for $\CP^{d-1}$ provide lower bounds on $\mathfrak{m}$.
The probably best non-embedding result in this case is by Mayer \cite{Mayer65} from which we now obtain the following (\cite{Mayer65}, Sec. 4.6):

\begin{theorem}[Pure state informational completeness -- lower bound] 
Informational completeness for the set of pure states in $\C^d$ requires
\be \mathfrak{m} > \left\{\begin{array}{ll}
                 2D-2\alpha & \forall d>1, \\
                 2D-2\alpha+2 &  d \mbox { odd, and }\alpha=3{\rm mod}4 \\ 
                 2D-2\alpha+1 & d \mbox { odd, and }\alpha=2{\rm mod}4,
               \end{array}\right.\label{eq:CPMayerbound66}
\ee
where $\alpha$ denotes the number of $1$'s in the binary expansion of $d-1$ and $D=2d-2$ is the real dimension of the manifold.
\end{theorem} 

If we now combine this non-embedding result with the upper bounds discussed in Subsection \ref{sec:pic}, we obtain a fairly comprehensive picture on the minimal number of measurements required to identify an element of the set of pure states. 
The bounds for the minimal number for the dimensions $2-10$ are summarized in the table below. 

\begin{figure}
\begin{center}
\begin{tabular}{ | c | c |  c |}
\hline
$d$ &  $4d-5$ & $\Min{\prem_1}$ \\
\hline
\hline
2  & 3 & 3 \\ 
3 & 7 & 7 \\ 
4 & 11 & 9 \textrm{or} 10\\ 
5 & 15 & 15 \\ 
6 & 19 & 17 \textrm{or} 18\\ 
7 & 23 & 22 \textrm{or} 23\\
8 & 27 & 23 -- 25\\ 
9 & 31 & 31 \\
10 & 35 & 33 \textrm{or} 34 \\
\hline
\end{tabular}
\caption{The bounds for the minimal number of operators for the dimensions $2-10$.}
\end{center}
\end{figure}

For all dimensions $d$ we can say that the difference between the upper and lower bounds is at most $\log_2(d)$, and that the minimal number differs from the best affine upper bound $4d-5$ by at most $2\log_2(d)$.
The lower and upper bounds up to $d=30$ are presented in Fig.\ref{fig:upperlowerbound}.

\begin{example}[Non-embedding results for Grassmannian manifolds]
For the $2r(d-r)$ dimensional complex Grasmannian manifold $G(r,d-r)$, i.e., states which are maximally mixed within a subspace of dimension $r$,
again Theorem \ref{thm:infoembeddingsmooth} and the subsequent Lemma assert that non-embedding results carry over to lower bounds on $\mathfrak{m}$. From \cite{Sugawara79} we obtain a bound for embeddings $G(r,d-r)\ra\R^\mathfrak{m}$ in the form of Eq.\eqref{eq:CPMayerbound66} where $D$ has to be set to $D=2r(d-r)$ and $\alpha$ has to be replaced by $\sum_{j=1}^r \beta(d-j)-\beta(j-1)$ where $\beta(n)$ is the number of ones in the binary expansion of $n$.  For $r=1$ this coincides with the aforementioned bound for pure states. Note that this provides a lower bound for informationally complete measurement schemes w.r.t. all sets which includes such a  ``Grassmannian manifold'', like the set of density matrices with rank bounded by $r$.
\end{example}

\begin{figure}
\includegraphics[width=0.7\textwidth]{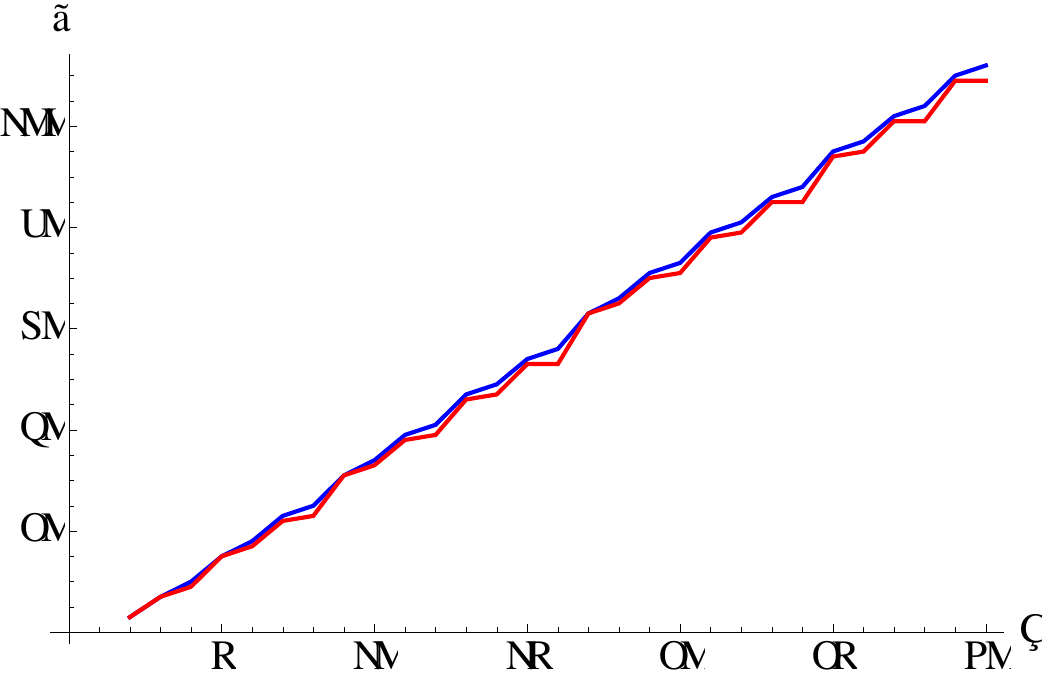}
\caption{\label{fig:upperlowerbound} 
The derived upper (blue) and lower (red) bounds for the minimal number $\mathfrak{m}$ for which informational completeness for pure states in $\C^d$ can be achieved. Note that the bounds differ by at most $3$ until $d=30$. The slope of the best affine upper bound is exactly 4.}
\end{figure}

\section{Summary}

How many measurement outcomes (i.e., POVM elements) are minimally needed in order to identify all quantum states from a given set? We have shown on the one hand that if the set is a manifold, then topological obstructions can increase the number of required measurements by a factor of two over the dimensionality of the manifold. On the other hand we have seen that this factor of two is sufficient even in a more general context where the considered set is not necessarily a manifold and its dimensionality is understood as its Minkowski dimension.

We have considered two different types of examples: 2-manifolds, where (non-)orientability plays an important role, and Grassmannian manifolds, which contain the set of all pure states as a particular instance.

For the latter case we have shown that upper and lower bounds---both obtained via topological embeddings---essentially match. In fact, they never differ by more than $\log_2(d)$. To be more precise, their difference is strictly less than the number of ones appearing in the binary expansion of $d-1$. 

Points which are beyond the present work, albeit obviously of practical relevance, are the inversion algorithm and issues of robustness and certifiability. 
For instance, we do not know whether in the pure state case $m\sim 4d$ can be achieved in a way such that (i) the inversion is algorithmically efficient and (ii) the validity of the assumption behind the prior information is certifiable.

\section*{Acknowledgements} 

We thank Claudio Carmeli for pointing out an error in the earlier version of Theorem \ref{prop:milgram}.

MMW was supported by the EU-STREP projects COQUIT and QUEVADIS and the Alfried Krupp von Bohlen und Halbach-Stiftung.
TH was supported by the Emil Aaltonen Foundation,  Alfred Kordelin Foundation and Academy of Finland (grant no. 1381359).
TH and MMW are both grateful for the hospitality and the inspiring
working environment at the Institut Mittag-Leffler in Stockholm, where this work was partly done.


\end{document}